\documentclass[11pt,preprint]{amsart}
\usepackage{amsfonts}
\usepackage{natbib}
\usepackage{graphicx}
\usepackage{color}
\usepackage{amssymb}
\usepackage{amsfonts}
\usepackage{amsthm}

\usepackage{hyperref}
\hypersetup{
    colorlinks=true,
    linkcolor=blue,
    filecolor=magenta,      
    urlcolor=cyan,
    citecolor=blue,
}

\newtheorem{theorem}{Theorem}[section]

\newtheorem{proposition}[theorem]{Proposition}

\newtheorem{example}[theorem]{Example}

\def\co{\mbox{co}}
\def\R{\mathbb R}

\newcommand{\bone}{\mathbf{1}}

\newcommand{\E}{\mathbb{E}}
\newcommand{\Prob}{\mathbb{P}}

\begin{document}
\bibliographystyle{plainnat}

\title[Allee effects in noisy environment]{A dynamical trichotomy for structured populations experiencing positive density-dependence in stochastic environments}
\author[S.J. Schreiber]{Sebastian J. Schreiber}
\email{sschreiber@ucdavis.edu}
\address{Department of Evolution and Ecology, One Shields Avenue, University of California, Davis, CA 95616 USA}

\begin{abstract} Positive density-dependence occurs when individuals experience increased survivorship, growth, or reproduction with increased population densities. Mechanisms leading to these positive relationships include mate limitation, saturating predation risk, and cooperative breeding and foraging. Individuals within these populations may differ in age, size, or geographic location and thereby structure these populations. Here, I study structured population models accounting for positive density-dependence and environmental stochasticity i.e. random fluctuations in the demographic rates of the population. Under an accessibility  assumption (roughly, stochastic fluctuations can lead to populations getting small and large), these models are shown to exhibit a dynamical trichotomy: (i) for all initial conditions, the population goes asymptotically extinct with probability one, (ii) for all positive initial conditions, the population persists and asymptotically exhibits unbounded growth, and (iii) for all positive initial conditions, there is a positive probability of asymptotic extinction and a complementary positive probability of unbounded growth. The main results are illustrated with applications to spatially structured populations with an Allee effect and age-structured populations experiencing mate limitation. 
\end{abstract}

\maketitle

\section{Introduction}

Higher population densities can increase the chance of mating success, reduce the risk of predation, and increase the frequency of cooperative behavior~\citep{courchamp2008allee}. Hence, survivorship, growth, and reproductive rates of individuals can exhibit a positive relationship with density i.e. positive density-dependence. In single species models, positive density-dependence can lead to an Allee effect: the existence of a critical density below which the population tends toward extinction and above which the population persists~\citep{dennis-89,mccarthy-97,scheuring-99,gascoigne-lipcius-04,tpb-03}. Consequently, the importance of Allee effects have been widely recognized for conservation of at risk populations and the management of invasive species~\citep{courchamp2008allee}. Population experiencing environmental stochasticity and a strong Allee effect are widely believed to be especially vulnerable to extinction as the fluctuations may drive their densities below the critical threshold~\citep{courchamp-etal-99}. When population densities lie above the critical threshold for the unperturbed system,  analyses and simulations of stochastic models support this conclusion~\citep{dennis-89,dennis-02,liebhold-bascompte-03,jbd-14,dennis-etal-15,assas-etal-16}. However, these studies also show that  when population densities lie below the critical threshold, stochastic fluctuations can rescue the population from the deterministic vortex of extinction. 

Individuals within population often differ in diversity of attributes including age, size, gender, and geographic location~\citep{caswell-01}. Positive density-dependence may differentially impact individuals in populations structured by these attributes~\citep{gascoigne-lipcius-05,courchamp2008allee}. This positive density-dependence can lead to an Allee threshold surface (usually a co-dimension one stable manifold of an unstable equilibrium) that separates population states that lead to extinction from those that lead to persistence~\citep{pams-04}. 

While several studies have examined how environmental stochasticity and population structure interact to influence persistence of populations experiencing negative-density dependence~\citep{hardin-etal-88a,tpb-09,jmb-14,hening-etal-16}, I know of no studies that examine this issue for populations experiencing positive density-dependence. To address this gap, this paper examines stochastic, single species models of the form 
\begin{equation}\label{eq:model}
X_{t+1}= A(X_t, \xi_{t+1}) X_t
\end{equation}
where $X_t=(X_{1,t},X_{2,t},\dots,X_{n,t})\in [0,\infty)^n$ is a column-vector of population densities, $A(X_t,\xi_{t+1})$ is a $n\times n$ non-negative matrix that determines the population densities in the next year as a function of the current densities $X_t$ and the environmental state $\xi_{t+1}$ over the time interval $[t,t+1)$. To focus on the effects of positive density-dependence, I assume that the entries of $A$ are non-decreasing functions of the population densities. Under additional suitable assumptions described in Sections 2 and 3, this paper shows that there is a dynamical trichotomy for \eqref{eq:model}: (i) asymptotic extinction occurs with probability one for all initial conditions, (ii) long-term persistence occurs with probability one for all positive initial conditions, or (iii) long-term persistence  and asymptotic extinction occur with complementary positive probabilities for all positive initial conditions. The model assumptions and definitions are presented in Section 2. Exemplar models of spatially-structured populations and age-structured populations are also presented in this section. The main results and applications to the exemplar models appear in Sections 3 and 4. Proofs of the main results are relegated to Section 5.

\section{Models, assumptions, and definitions}
Throughout this paper, I consider stochastic difference equations of the form given by equation~\eqref{eq:model}. The state space for these equations is the non-negative cone $C=[0,\infty)^n$. Define the standard ordering on this cone by $x\ge y$ for $x,y\in C$ if $x_i\ge y_i$ for all $i$. Furthermore, $x>y$ if $x\ge y$ but $x\neq y$ and $x\gg y$ if $x_i>y_i$ for all $i$. Throughout, I will use $\|x\|=\max_i|x_i|$ to denote the sup norm and $\|A\|=\max_{\|x\|=1} \|Ax\|$ to denote the associated operator norm. Define the co-norm of a matrix $A$ by 
$\co(A)=\min_{\|x\|=1}\|Ax\|$. The co-norm is the minimal amount that the matrix $A$ stretches a vector. Define $\log^+x=\max\{\log x, 0\}$ to be the non-negative component of $\log x.$

For \eqref{eq:model}, there are \emph{five standing assumptions} 
\begin{description}
\item [A1.Uncorrelated environmental fluctuations]  $\{\xi_t\}_{t=0}^\infty$ is a sequence of independent and identically distributed (i.i.d) random variables taking values in a separable metric space $E$ (such as $\R^k$).
\item [A2.Feedbacks depend continuously on population and environmental state] the entries of the matrix function $A_{ij}:C\times E \to [0,\infty)$ are continuous functions of population state $x$ and the environmental state $\xi$.
\item [A3.The population only experiences positive feedbacks] For all $i,j$ and $\xi\in E$, $A_{ij}(x,\xi)\ge A_{ij}(y,\xi)$ whenever $x\ge y.$
\item [A4.Primitivity] There exists $\tau\ge 0$ such that $A(x,\xi)^\tau\gg 0$ for all $x\gg0$ and $\xi\in E$. 
\item [A5.Finite logarithmic moments]  For all $c\ge 0$, $\E[\log^+\|A(c\bone,\xi_t)\|]<\infty$. There exists $c^*>0$ such that $\E[\log^+ (1/\co\left(\prod_{t=1}^\tau A(c\bone,\xi_t)\right))]<\infty$ for all $c\ge c^*$.
\end{description}

The first assumption implies that $(X_t)_{t\ge 0}$ is a Markov chain on $C$ and the second assumption ensures this stochastic process is Feller. The third assumption is consistent with the intent of understanding how non-negative feedbacks, in and of themselves, influence structured population dynamics. An important implication of this assumption is that the system is monotone i.e. if $X_0> \widetilde X_0>0$, then $X_t \ge \widetilde X_t$ for all $t\ge \tau$ where $X_t,\widetilde X_t$ are solutions to \eqref{eq:model} with initial conditions $X_0$ and $\widetilde X_0$, respectively. The fourth assumption ensures that all states in the population contribute to all other population states after $\tau$ time steps. The final assumption is meet for most models and ensures that \citet{kingman-73}'s subadditive ergodic theorem and the random Perron-Frobenius theorem of \citet{arnold-etal-94} are applicable. 

To see that these assumption include  models of biological interest, here are a few examples. 

\subsection*{Example 1 (Scalar models)} Considered an unstructured population with $n=1$ in which case $x\in [0,\infty)$. To model mate limitation, \citet{mccarthy-97} considered a model where $x$ corresponds to the density of females and, with the assumption of a 1:1 sex ratio, also equals the density of males. The probability of a female successfully mating is given by $ax/(1+ax)$ where $x$ is the male density and $a>0$ determines how effectively individuals find mates. If a mated individual produces on average $\xi$ daughters, then the population density in the next year is $\xi ax^2/(1+sx)$. If we allow $\xi$ to be stochastic, then \eqref{eq:model} is determined by $A(x,\xi)=\xi a x^2/(1+a x)$. Allowing the $\xi_t$ to be a log-normal would satisfy assumptions \textbf{A1}--\textbf{A5}. 

To model predator saturation~\citep{tpb-03}, let $\exp(-M/(1+hx))$ be the probability that individual escapes predation from a predator population with an ``effective'' attack rate of $M$ and handling time $h$. If $\xi$ is the number of offspring produced by an individual which escaped predation, then the population density in the next year is $\xi x \exp(-M/(1+hx))$.  Letting $\xi$ be stochastic yields $A(x,\xi)=\xi x \exp(-M/(1+h x))$. Allowing the $\xi_t$ to be a log-normal would satisfy assumptions \textbf{A1}--\textbf{A5}. 

Finally, \citet{liebhold-bascompte-03} used a more phenomenological model of the form $A(x,\xi)=\exp(x-C+\xi)$ where $C$ is the critical threshold in the absence of stochasticity and $\xi$ are normally distributed with mean zero. This model also satisfies all of the assumptions. 

We can use these scalar models, which were studied by \citep{jbd-14}, to build structured models as the next two examples illustate.  

\subsection*{Example 2 (Spatial models)} Consider a population that live in $n$ distinct patches. $x_i$ is the population density in patch $i$. Let $C_i>0$ be the critical threshold in patch $i$ and $\xi_i$ be the environmental state in patch $i$. Let $d_{ij}$ be the fraction of individuals dispersing from patch $j$ to patch $i$, and $D=(d_{ij})$ be the corresponding dispesal matrix. Then the spatial model is
\begin{equation}\label{eq:spatial}
A(x,\xi)=D \mbox{diag}(\exp(x_1-C_i+\xi_1),\exp(x_2-C_2+\xi_2),\dots, \exp(x_n-C_n+\xi_n))
\end{equation}
where $\mbox{diag}$ denotes a diagonal matrix with the indicated diagonal elements. If $D$ is a primitive matrix and the $\xi_{t}=(\xi_{1,t},\dots,\xi_{n,t})$ are a multivariate normals with zero means, then this model satisfies the assumptions.

\subsection*{Example 3 (Age-structured models)} Consider a population with $n$ age classes and $x_i$ is the density of age $i$ individuals. Assume that final $\ell$ age classes reproduce i.e. ages $n-\ell+1,n-\ell+2,\dots,n$ reproduce. If mate limitation causes positive density dependence (see Example 1) and reproductively mature individuals mate randomly, then  the fecundity of individuals in age class $n-\ell+i$ equals $f_i(x,\xi)=\xi_i \sum_{j=n-\ell+1}^n x_j/(1+\sum_{j=n-\ell+1}^n x_j)$ where $\xi_i$ is the maximal fecundity of individuals of age $i$. Let $s_i$ be the probability an individual survives from age $i-1$ to age $i$. This yields the followng nonlinear Leslie matix model
\begin{equation}\label{eq:age}
A(x,\xi)=\begin{pmatrix}0&\dots&0&f_{1}(x,\xi)&\dots & f_\ell(x,\xi)\\
s_2&0&0&\dots& 0&0 \\
0&s_3&0&\dots&0& 0 \\
\vdots&\vdots&\vdots&\vdots&\vdots&\vdots\\
0&0&\dots&0&s_n& 0 \\
\end{pmatrix}.
\end{equation}
If $\ell\ge 2$ and $\xi_t=(\xi_{1,t},\dots,\xi_{n,t})$ are multivariate log-normals, then this model satisfies the assumptions \textbf{A1}--\textbf{A5}.

\section{Main results}

To state the main results, consider the linearization of \eqref{eq:model} at the origin and near infinity. At the origin, the linearized dynamics are given by $X_{t+1}=A(0,\xi_{t+1})X_t$. Hence, the rate at which the population grows at low density is approximately given by the rate at which the random product of matrices, $A(0,\xi_t)\dots A(0,\xi_1)$, grows. \citet{kingman-73}'s subadditive ergodic theorem implies there exists $r_0$ (possibly $-\infty$) such that 
\[
\lim_{t\to\infty}\frac{1}{t}\log \|A(0,\xi_t)\dots A(0,\xi_1)\| =r_0 \mbox{ with probability one.}
\]
To characterize population growth near infinity, for all $c>0$ the subadditive ergodic theorem implies there exists an $r_c$ such that 
\[
\lim_{t\to\infty}\frac{1}{t}\log \|A(c\bone,\xi_t)\dots A(c\bone,\xi_1)\| =r_c \mbox{ with probability one}
\]
where $\bone=(1,1,\dots,1)$ is the vector of ones. Due to our assumption that the entries of $A(x,\xi)$ are non-decreasing with respect the entries of $x$, $r_c$ is non-decreasing with respect to $c$. Hence, the following limit exists (possibly $+\infty$)
\[
r_\infty = \lim_{c\to\infty} r_c.
\]

With these definitions and assumptions, the following theorem is proven in Section 5. 

\begin{theorem}\label{thm:local} 
\begin{description}
\item[Unconditional persistence] If $r_0>0$, then 
\[
\lim_{t\to\infty} \|X_{t}\|=\infty \mbox{ with probability one whenever $X_0\gg 0$.}
\]
\item[Unconditional extinction]
If $r_\infty<0$, then 
\[
\lim_{t\to\infty} X_{t}=0 \mbox{ with probability one.}
\]
\item[Conditional persistence and extinction]
If $r_0<0<r_\infty$, then for all $\varepsilon>0$ there exist $c^*>c_*>0$ such that 
\[
\Prob\left[\lim_{t\to\infty} X_t=0\Big|X_0=x \right]\ge 1-\varepsilon \mbox{ whenever }x\le c_*\bone
\]
and
\[
\Prob\left[\lim_{t\to\infty}\| X_t\|=\infty\Big|X_0=x \right]\ge 1-\varepsilon \mbox{ whenever }x\ge c^*\bone.
\]
\end{description}
\end{theorem}

To get statements about all initial conditions with probability one in the final case, an assumption that ensures that the environmental stochasticity can drive the population to low or high densities is needed. Define $\{0,\infty\}$ to be \emph{accessible} if for all $c>0$ there exists $\gamma>0$ such that
\[
\Prob\left[\left\{\mbox{there is }t\ge0 \mbox{ such that }X_t\gg c\bone \mbox{ or } X_t \ll \bone/c\right\}\Big|X_0=x\right]\ge \gamma
\]
for all $x\gg 0.$ All of the examples in Section 2 satisfy this accessibility condition. 

\begin{theorem}\label{thm:global}
If $r_0<0<r_\infty$ and $\{0,\infty\}$ is accessible, then
\[
\Prob\left[\lim_{t\to\infty}\|X_{t}\|=\infty \mbox{ or }
\lim_{t\to\infty}X_t=0\Big|X_0=x \right]= 1.\]
\end{theorem}

Proofs of both theorems are presented in Section 5. The scalar version of these theorems were proven in Theorem 3.2 of \citep{jbd-14}. 

\section{Applications}

To illustrate the applicability of the two theorems, we consider the spatial structured and age structured models introduced in section 3. 

\subsection*{Example 2 (spatially structured populations) revisited} Consider the spatial structured model described in Example 2 and characterized by \eqref{eq:spatial}. For this model, 
\[
A(c\bone,\xi)=D \mbox{diag}(\exp(-C_1+\xi_1),\exp(-C_2+\xi_2),\dots, \exp(-C_n+\xi_n))\exp(c).
\]
For simplicity, let us assume that the fraction of individuals dispersing is $d$ and dispersing individuals land with equal likelihood on any patch (including the possibility of returning to its original patch). Then 
$D=(d_{ij})$ is given by $d_{ij}=d/n$ for $i\neq j$ and $d_{ii}=(1-d)+d/n$. Assume that $d\in (0,1].$

I claim that $r_\infty=\infty$. Indeed, let $b=\max\{1-d,d/k\}>0.$ Then $D\ge b\mbox{Id}$ where $\mbox{Id}$ denotes the identity matrix and 
\begin{eqnarray*}
\E[\log \| \prod_{s=1}^t A(c\bone, \xi_s)\|]&\ge&
\E[\log\| \prod_{s=1}^t b \mbox{diag}(\exp(-C_1+\xi_{1,s}),\exp(-C_2+\xi_{2,s}),\dots, \exp(-C_n+\xi_{n,s}))\exp(c) \| ]\\
&\ge& \E[\log\| \prod_{s=1}^t \mbox{diag}(\exp(\xi_{1,s}),\exp(\xi_{2,s}),\dots, \exp(\xi_{n,s})) \| ]+t( c+ \log b-  \max_i C_i)\\
&=&  \E[\max_i \sum_{s=1}^t \xi_{i,s} ]+t( c+ \log b-  \max_i C_i)\\
&\ge& t\left(\E[\xi_{1,1}]+c+\log b-  \max_i C_i\right).
\end{eqnarray*}
Dividing by $t$ and taking the limit as $t\to\infty$, this inequality implies that $r_c\ge \E[\xi_{1,1}]+c+\log b-\max_i C_i$. Hence, $r_\infty=\lim_{c\to\infty} r_c = \infty$ as claimed. Theorem~\ref{thm:global} implies that for all $x\gg 0$, $X_t\to\infty$ with positive probability whenever $X_t=x.$

Understanding $r_0$ is more challenging. However, Proposition 3 of \citet{tpb-09} implies that $r_0$ varies continuously as a function of $d$. In the limit of $d=0$, $D=\mbox{Id}$ and $r_0=\max_i \E[\xi_{i,1}]$. Hence, for populations where $d\approx 0$ but $d>0$, there are two types of dynamics. If $\E[\xi_{i,1}]<0$ for all patches (i.e. populations are unable to persist in each patch at low density), then  there is a positive probability of going either asymptotically extinct or a complementary positive probability of persistence. Alternatively, if $\E[\xi_{i,1}]>0$ for at least one patch, then the population persists with probability one whenever $X_0\gg 0$. 

Now consider the case that all individuals disperse i.e.  $d=1$. Then $r_0=\E[\log \frac{1}{n} \sum_i \exp(\xi_{i,1})]$ i.e. $e^{r_0}$ is the geometric mean of the spatial average of the $\exp(\xi_{i,1})$. By Jensen's inequality, $r_0$ when $d\approx 1$ is greater than $r_0$ when $d\approx 0$. Hence, one can get the scenario where increasing the dispersal fraction $d$ shifts a population from experiencing asymptotic extinction with positive probability to a population that persists with probability one. This corresponds to a positive density-dependence analog of a ph,enomena observed in models with negative density-dependent feedbacks~\citep{tpb-09,hening-etal-16} and density-independent feedbacks~\citep{metz-etal-83,jansen-yoshimura-98,prsb-10,jmb-13}. However, in these models, the long-term outcome never exhibits a mixture of extinction and persistence.

\subsection*{Example 3 (age-structured populations) revisited}
Consider the age-structured model with mate-limitation in Example 3 where there are $\ell\ge 2$ reproductive stages. If $\xi_t$ are multivariate log-normals, then $\{0,\infty\}$ is accessible. Define 
\[
B=\begin{pmatrix}0&0&0&0&\dots & 0\\
s_2&0&0&\dots& 0&0 \\
0&s_3&0&\dots&0& 0 \\
\vdots&\vdots&\vdots&\vdots&\vdots&\vdots\\
0&0&\dots&0&s_n& 0 \\
\end{pmatrix}.
\]
As $0<s_i<1$ for all $i$, the dominant eigenvalue $\lambda$ of $B$ is strictly less than one. Thus,
\begin{eqnarray*}
r_0=\lim_{t\to\infty}\frac{1}{t}\E[\log \| \prod_{s=1}^t \|A(0,\xi_s) \|]&=&\lim_{t\to\infty}\frac{1}{t}\log \| B^t \|\\
&=& \log \lambda<0.
\end{eqnarray*}
As $r_0<0$, it follows that for all positive initial conditions there is a positive probability of asymptotic extinction (in contrast the spatial model which always has a positive probability of persistence and unbounded growth.) 

To say something about persistence,  assume that $\xi_{1,t},\dots,\xi_{\ell,t}$ have the same log mean $\mu$ and non-degenerate log-covariance matrix $\Sigma^2.$ Then $r_\infty$ is an increasing function of $\mu$ with $\lim_{\mu\to\infty} r_\infty=\infty$ and $\lim_{\mu\to-\infty} r_\infty<0$. Hence, there is a critical $\mu$, call it $\mu^*$, such that the population goes asymptotically extinct with probability one whenever $\mu<\mu^*$ and the population persists with positive probability whenever $\mu>\mu^*$.

\section{Proofs}

First, I prove Theorem~\ref{thm:local}. Assume $r_0>0$ and $X_0=x_0\gg0$. As the entries of $A(x,\xi)$ are non-decreasing functions of $x$, 
\begin{eqnarray*}
\liminf_{t\to\infty}\frac{1}{t}\log\|X_t\|&=& \liminf_{t\to\infty}\frac{1}{t}\log\|\prod_{s=1}^t A(X_{s-1},\xi_s)x_0 \|\\
&\ge & \liminf_{t\to\infty}\frac{1}{t}\log\|\prod_{s=1}^t A(0,\xi_s)x_0 \|\\
&=&r_0>0 \mbox{ with probability one.}\\
\end{eqnarray*}
In particular, $\lim_{t\to\infty}\|X_t\|=\infty$ with probability one as claimed.

Next, assume that $r_\infty<0$. Given any $X_0=x_0\gg0$, choose $c>0$ such that $c\mathbf{1}\ge x_0$ and 
\[
\lim_{t\to\infty}\frac{1}{t}\log\|\prod_{s=1}^t A(c\mathbf{1},\xi_{s})\|\le r_\infty/2 \mbox{ with probability one.}
\]
Then 
\begin{eqnarray*}
\limsup_{t\to\infty}\frac{1}{t}\log\|X_t\|&\le & \limsup_{t\to\infty}\frac{1}{t}\log\|\prod_{s=1}^t A(c\mathbf{1},\xi_s)x_0 \|\\
&\le & r_0/2<0 \mbox{ with probability one.}\\
\end{eqnarray*}
In particular, $\lim_{t\to\infty} X_t=0$ with probability one as claimed. 

Finally, assume that $r_\infty>0$ and $r_0<0.$ As the entries of $A$ are non-increasing in $x$, there exists $c>0$ such that $A(c\mathbf{1},\xi)\le A(0,\xi)\exp(-r_0/2)$ for $\xi\in E$. Hence, 
\begin{equation}\label{eq:ub}
\limsup_{t\to\infty} \frac{1}{t}\log \|\prod_{s=1}^tA(c\mathbf{1},\xi_s)\|
\le r_0/2 <0\mbox{ with probability one.}
\end{equation}
Define the random variable
\[
R = \sup_{t\ge 1} {\|\prod_{s=1}^t A(c\mathbf{1},\xi_s) \|}.
\]
Equation~\ref{eq:ub} implies that $R<\infty$ with probability one. For all $k>0$, define the event $\mathcal{E}_k=\{R\le k\}.$ For $x_0\le c\bone/k$ and $X_0=x_0$, I claim that $X_t\le c \mathbf{1}$ for all $t\ge 0$ on the event $\mathcal{E}_k.$ I prove this claim by induction. $X_0\le  c\bone$ by assumption. Suppose that $X_s\le c\bone$ for $0\le s\le t-1$. Then 
\begin{eqnarray*}
\|X_{t}\|&=&\|\prod_{s=1}^tA(X_{s-1},\xi_{s})x_0\|\\
&\le &\|\prod_{s=1}^t A(c\mathbf{1},\xi_{s})c\mathbf{1}/k\| \mbox{ by induction and monotonicity}\\
&\le &\|\prod_{s=1}^t A(c\mathbf{1},\xi_{s})\| c/k \le  Rc/k \mbox{ by the definition of }R\mbox{ and }x\\
&\le& c \mbox{ on the event }\mathcal{E}_k.
\end{eqnarray*}
This completes the proof of the claim that $X_t \le c \bone$ for all $t\ge 0$ on the event $\mathcal{E}_k.$ It follows that on the event $\mathcal{E}_k$ and $X_0=x\le c\mathbf{1}$ that
\begin{eqnarray*}
\limsup_{t\to\infty}\frac{1}{t}\log \| X_t\| &\le& \limsup_{t\to\infty} \frac{1}{t}\log \|\prod_{s=1}^t  A(c\mathbf{1},\xi_s)\|c\\
&\le & r_0/2<0 \mbox{ almost surely.}
\end{eqnarray*}
In particular, $\lim_{t\to\infty}X_t =0$ almost sure on the event $\mathcal{E}_k$. As the events $\mathcal{E}_k$ are increasing with $k$,  $\lim_{k\to\infty}\Prob[\mathcal{E}_k]=\Prob[\cup_k \mathcal{E}_k]=\Prob[R<\infty]=1.$ Therefore, given $\varepsilon>0$, there exists $k$ such that $\Prob[\mathcal{E}_k]>1-\varepsilon$. For this $k$, $x_0\le c \mathbf{1}/k$ and $X_0=x_0$, 
\[
\Prob[\lim_{t\to\infty}X_t=0|X_0=x_0]\ge \Prob[\mathcal{E}_k]\ge 1-\varepsilon. 
\]

To show convergence to $\infty$ with positive probability when $r_\infty>0$, choose $c\ge c^*$ sufficiently large so that 
\[
\lim_{t\to\infty}\frac{1}{t}\log \|\prod_{s=1}^t A(c\mathbf{1},\xi_s)\| \ge r_\infty/2>0 \mbox{ with probability one.}
\]
By the Random Perron-Frobenius theorem~\citep[Theorem 3.1 and Remark (ii) on pg. 878]{arnold-etal-94}, 
\begin{equation}\label{eq:c-large}
\lim_{t\to\infty}\frac{1}{t}\log \left(e_i^T\prod_{s=1}^t A(c\mathbf{1},\xi_s)e_j\right) \ge r_\infty/2>0 \mbox{ with probability one.}
\end{equation}
for all elements $e_i,e_j$ of the standard basis of $\R^n$ and where $^T$ denotes the transpose of a vector. Equation~\ref{eq:c-large} implies that all of the entries of $\prod_{s=1}^t A(c\mathbf{1},\xi_s)$ grow exponentially in time at rate greater than $r_\infty/2$ with probability one. 

Define 
\begin{eqnarray*}
R_\infty&=&
\inf_{t\ge 1,1\le i\le n} e_i^T \prod_{s=1}^t A(c\mathbf{1},\xi_s)c\bone.
\end{eqnarray*}
By \eqref{eq:c-large} and the primitivity assumption \textbf{A4}, $R_\infty>0$  with probability one. Define the events
\[
\mathcal{F}_k=\{R_\infty>1/k\}\mbox{ for }k\ge 1.
\]
Now, suppose that $X_0=x_0\ge c\mathbf{1}k$. I claim that $X_{t}\ge c\mathbf{1}$ for all $t\ge 0$ on the event $\mathcal{F}_k$. $X_0\ge c\mathbf{1}$ by the choice of $x_0$. Assume that $X_s\ge c\bone$ for $0 \le s \le t-1$. Then 
\begin{eqnarray*}
X_{t} &=& \prod_{s=1}^{t} A(X_{s-1},\xi_s)x_0 \\
&\ge & \prod_{s=1}^{t} A(c\mathbf{1},\xi_s) x_0 \mbox{ by inductive hypothesis}\\
&\ge& R_\infty c\mathbf{1}k \mbox{ by definition of }R_\infty\mbox{ and }x_0\\
&\ge& c\mathbf{1}\mbox{ on the event }\mathcal{F}_k.
\end{eqnarray*}
Equation \eqref{eq:c-large} implies that on the event $\mathcal{F}_k$
\[
\liminf_{t\to\infty}\frac{1}{t}\log \| X_t\| \ge r_0/2 \mbox{ almost surely.}
\]
Hence, $\lim_{t\to\infty} \|X_t\|=\infty$ almost surely on the event $\mathcal{F}_k.$ As $\mathcal{F}_k$ are an increasing set of events, $\Prob[R_\infty>0]=\Prob[\cup_{t\ge 1}\mathcal{F}_k]=1$. For any $\varepsilon>0$ there is $k\ge 1$ such that $\Prob[\mathcal{F}_k]\ge 1-\varepsilon.$ Hence, for this $k$ and $X_0=x\ge c k \mathbf{1}$, 
\[
\Prob[\lim_{t\to\infty}\|X_t\|=\infty| X_0=x]\ge 1-\varepsilon.
\]
This completes the proof of Theorem~\ref{thm:local}.

The proof of Theorem~\ref{thm:global} follows from Theorem~\ref{thm:local} and  the following proposition. 

\begin{proposition}\label{general-global-result} Assume $\{0,\infty\}$ is accessible. Let $c>0$ and $\delta\in [0,1)$ be such that 
\[
\Prob\left[ \lim_{t\to\infty} X_t=0|X_0=x \right]\ge 1-\delta \mbox{ whenever }x\le \bone/c
\]
and 
\[
\Prob\left[ \lim_{t\to\infty} X_t=\infty|X_0=x \right]\ge 1-\delta \mbox{ whenever }x\ge c\bone.
\]
Then
\[
\Prob\left[\lim_{t\to\infty} X_t=\infty \mbox{ or } \lim_{t\to\infty} X_t =0 | X_0=x\right]=1 \mbox{ whenever }x\gg 0. 
\]
\end{proposition}
\begin{proof}
Define the event \[\mathcal{C}=\left\{\lim_{t\to\infty} X_t=\infty \mbox{ or } \lim_{t\to\infty} X_t =0\right\}.\] For any $x\in C$, define $\Prob_x[\mathcal{E}]=\Prob[\mathcal{E}|X_0=x]$ for any event $\mathcal E$ in the $\sigma$-algebra generated by $\{X_0=x,X_1,X_2,\dots\}$. Define the stopping time 
\[
S = \inf \{t\ge 0 \ : X_t \ge c \bone \mbox{ or } X_t \le \bone/c\}.
\]  Since $\{0,\infty\}$ is accessible, there exists $\gamma >0$ such that $\Prob_{x}[S<\infty]>\gamma$ for all $x\gg 0$. Let $I_{\{ S<\infty\}}$ equal $1$ if $S<\infty$ and $0$ otherwise. The strong Markov property implies that for all $x\gg 0$
\begin{eqnarray*}
\Prob_x\left[ \mathcal{C}\right] &=& \E_{x}\left[ \Prob_{X_S} \left[ \mathcal{C} \right] I_{\{ S<\infty\}} \right]+\E_{x}\left[ \Prob_{X_S} \left[ \mathcal{C} \right] I_{\{ S=\infty\}} \right]\\
&=& \E_{x}\left[ \Prob_{X_S} \left[ \mathcal{C} \right] I_{\{ S<\infty\}} \right]\\
&\ge&(1-\delta)\gamma.
\end{eqnarray*} 
Let $\mathcal{F}_t$ be the $\sigma$-algebra generated by $\{X_1,\dots,X_t\}$.
The L\'{e}vy zero-one law implies that for all $x\gg 0$, $\lim_{t\to \infty} \E_{x}\left[ I_{\mathcal{C}} | \mathcal{F}_t\right]= I_{\mathcal{C}}$ almost surely. On the other hand, the Markov property implies that $\E_{x}\left[ I_{\mathcal{C}} | \mathcal{F}_t\right]= \Prob_{X_t}[\mathcal{C}] \ge (1-\delta)\gamma$ for all $x\gg 0$. Hence $\Prob_{x}[\mathcal{C}]=1$ for all $x\gg0$. 

\end{proof}

\bibliography{allee}

\begin{thebibliography}{25}
\providecommand{\natexlab}[1]{#1}
\providecommand{\url}[1]{\texttt{#1}}
\expandafter\ifx\csname urlstyle\endcsname\relax
  \providecommand{\doi}[1]{doi: #1}\else
  \providecommand{\doi}{doi: \begingroup \urlstyle{rm}\Url}\fi

\bibitem[Arnold et~al.(1994)Arnold, Gundlach, and Demetrius]{arnold-etal-94}
L.~Arnold, V.~M. Gundlach, and L.~Demetrius.
\newblock Evolutionary formalism for products of positive random matrices.
\newblock \emph{Annals of Applied Probability}, 4:\penalty0 859--901, 1994.

\bibitem[Assas et~al.(2016)Assas, Dennis, Elaydi, Kwessi, and
  Livadiotis]{assas-etal-16}
L.~Assas, B.~Dennis, S.~Elaydi, E.~Kwessi, and G.~Livadiotis.
\newblock A stochastic modified beverton--holt model with the allee effect.
\newblock \emph{Journal of Difference Equations and Applications}, 22\penalty0
  (1):\penalty0 37--54, 2016.

\bibitem[Bena\"{i}m and Schreiber(2009)]{tpb-09}
M.~Bena\"{i}m and S.~J. Schreiber.
\newblock Persistence of structured populations in random environments.
\newblock \emph{Theoretical Population Biology}, 76:\penalty0 19--34, 2009.

\bibitem[Caswell(2001)]{caswell-01}
H.~Caswell.
\newblock \emph{Matrix Population Models}.
\newblock Sinauer, Sunderland, Massachuesetts, 2001.

\bibitem[Courchamp et~al.(1999)Courchamp, Clutton-Brock, and
  Grenfell]{courchamp-etal-99}
F.~Courchamp, T.~Clutton-Brock, and B.~Grenfell.
\newblock Inverse density dependence and the {A}llee effect.
\newblock \emph{Trends in Ecology and Evolution}, 14:\penalty0 405--410, 1999.

\bibitem[Courchamp et~al.(2008)Courchamp, Berec, and
  Gascoigne]{courchamp2008allee}
F.~Courchamp, L.~Berec, and J.~Gascoigne.
\newblock Allee effects in ecology and conservation.
\newblock \emph{Environ. Conserv}, 36\penalty0 (1):\penalty0 80--85, 2008.

\bibitem[Dennis(1989)]{dennis-89}
B.~Dennis.
\newblock Allee effects: {P}opulation growth, critical density, and the chance
  of extinction.
\newblock \emph{Natural Resources Modeling}, 3:\penalty0 481--538, 1989.

\bibitem[Dennis(2002)]{dennis-02}
B.~Dennis.
\newblock Allee effects in stochastic populations.
\newblock \emph{Oikos}, 96:\penalty0 389--401, 2002.

\bibitem[Dennis et~al.(2015)Dennis, Assas, Elaydi, Kwessi, and
  Livadiotis]{dennis-etal-15}
B.~Dennis, L.~Assas, S.~Elaydi, E.~Kwessi, and G.~Livadiotis.
\newblock Allee effects and resilience in stochastic populations.
\newblock \emph{Theoretical Ecology}, pages 1--13, 2015.

\bibitem[Evans et~al.(2013)Evans, Ralph, Schreiber, and Sen]{jmb-13}
S.~N. Evans, P.~Ralph, S.~J. Schreiber, and A.~Sen.
\newblock Stochastic growth rates in spatio-temporal heterogeneous
  environments.
\newblock \emph{Journal of Mathematical Biology}, 66:\penalty0 423--476, 2013.

\bibitem[Gascoigne and Lipcius(2005)]{gascoigne-lipcius-05}
J.~Gascoigne and R.N. Lipcius.
\newblock Periodic dynamics in a two-stage allee effect model are driven by
  tension between stage equilibria.
\newblock \emph{Theoretical population biology}, 68\penalty0 (4):\penalty0
  237--241, 2005.

\bibitem[Gascoigne and Lipcius(2004)]{gascoigne-lipcius-04}
J.C. Gascoigne and R.N. Lipcius.
\newblock Allee effects driven by predation.
\newblock \emph{Journal of Applied Ecology}, 41:\penalty0 801--810, 2004.

\bibitem[Hardin et~al.(1988)Hardin, Tak{\'a}{\v{c}}, and Webb]{hardin-etal-88a}
D.~P. Hardin, P.~Tak{\'a}{\v{c}}, and G.~F. Webb.
\newblock Asymptotic properties of a continuous-space discrete-time population
  model in a random environment.
\newblock \emph{Journal of Mathematical Biology}, 26:\penalty0 361--374, 1988.

\bibitem[Hening et~al.(2016)Hening, Nguyen, and Yin]{hening-etal-16}
A.~Hening, D.H. Nguyen, and G.~Yin.
\newblock Stochastic population growth in spatially heterogeneous environments:
  The density-dependent case.
\newblock \emph{arXiv preprint arXiv:1605.02027}, 2016.

\bibitem[Jansen and Yoshimura(1998)]{jansen-yoshimura-98}
V.~A.~A. Jansen and J.~Yoshimura.
\newblock Populations can persist in an environment consisting of sink habitats
  only.
\newblock \emph{Proceeding of the National Academy of Sciences USA},
  95:\penalty0 3696--3698, 1998.

\bibitem[Kingman(1973)]{kingman-73}
J.~F.~C. Kingman.
\newblock Subadditive ergodic theory.
\newblock \emph{Ann. Prob.}, 1:\penalty0 883--909, 1973.

\bibitem[Liebhold and Bascompte(2003)]{liebhold-bascompte-03}
A.~Liebhold and J.~Bascompte.
\newblock The {A}llee effect, stochastic dynamics and the eradication of alien
  species.
\newblock \emph{Ecology Letters}, pages 133--140, 2003.

\bibitem[McCarthy(1997)]{mccarthy-97}
M.~A. McCarthy.
\newblock The {A}llee effect, finding mates and theoretical models.
\newblock \emph{Ecological Modelling}, 103:\penalty0 99--102, 1997.

\bibitem[Metz et~al.(1983)Metz, de~Jong, and Klinkhamer]{metz-etal-83}
J.~A.~J. Metz, T.~J. de~Jong, and P.~G.~L. Klinkhamer.
\newblock What are the advantages of dispersing; a paper by {K}uno extended.
\newblock \emph{Oecologia}, 57:\penalty0 166--169, 1983.

\bibitem[Roth and Schreiber(2014{\natexlab{a}})]{jbd-14}
G~Roth and S.J. Schreiber.
\newblock Pushed to brink: Allee effects, environmental stochasticity, and
  extinction", special issue on allee effects for.
\newblock \emph{Journal of Biological Dynamics}, 8:\penalty0 187--205,
  2014{\natexlab{a}}.

\bibitem[Roth and Schreiber(2014{\natexlab{b}})]{jmb-14}
G.~Roth and S.J. Schreiber.
\newblock Persistence in fluctuating environments for interacting structured
  populations.
\newblock \emph{Journal of Mathematical Biology}, 68:\penalty0 1267--1317,
  2014{\natexlab{b}}.

\bibitem[Scheuring(1999)]{scheuring-99}
I.~Scheuring.
\newblock Allee effect increases dynamical stability in populations.
\newblock \emph{Journal of Theoretical Biology}, 199:\penalty0 407--414, 1999.

\bibitem[Schreiber(2003)]{tpb-03}
S.~J. Schreiber.
\newblock Allee effects, chaotic transients, and unexpected extinctions.
\newblock \emph{Theoretical Population Biology}, 2003.

\bibitem[Schreiber(2004)]{pams-04}
S.~J. Schreiber.
\newblock On {A}llee effects in structured populations.
\newblock \emph{Proc. Amer. Math. Soc.}, 132\penalty0 (10):\penalty0 3047--3053
  (electronic), 2004.
\newblock ISSN 0002-9939.

\bibitem[Schreiber(2010)]{prsb-10}
S.J. Schreiber.
\newblock Interactive effects of temporal correlations, spatial heterogeneity,
  and dispersal on population persistence.
\newblock \emph{Proceedings of the Royal Society: Biological Sciences},
  277:\penalty0 1907--1914, 2010.

\end{thebibliography}

\end{document}